\documentclass[12pt,english]{svjour3}
\usepackage[T1]{fontenc}
\usepackage[latin9]{inputenc}
\usepackage{babel}
\usepackage{units}
\usepackage{mathrsfs}
\usepackage{url}
\usepackage{amsmath}
\usepackage{amssymb}
\usepackage[unicode=true]
 {hyperref}

\makeatletter
\usepackage[vlined,plain,noresetcount]{algorithm2e} 
\SetAlgoSkip{smallskip}
\usepackage{multicol}

\usepackage{transparent}
\usepackage[adobe-utopia]{mathdesign}

\usepackage[margin=1in]{geometry}

\makeatother

\begin{document}

\title{On the Properties of the Priority Deriving Procedure in the Pairwise
Comparisons Method}

\author{Konrad Ku\l{}akowski}

\institute{AGH University of Science and Technology, \\
al. Mickiewicza 30, Kraków, Poland, \\
\href{mailto:konrad.kulakowski@agh.edu.pl}{konrad.kulakowski@agh.edu.pl}}
\maketitle
\begin{abstract}
The pairwise comparisons method is a convenient tool used when the
relative order of preferences among different concepts (alternatives)
needs to be determined. There are several popular implementations
of this method, including the Eigenvector Method, the Least Squares
Method, the Chi Squares Method and others. Each of the above methods
comes with one or more inconsistency indices that help to decide whether
the consistency of input guarantees obtaining a reliable output, thus
taking the optimal decision. 

This article explores the relationship between inconsistency of input
and discrepancy of output. A global ranking discrepancy describes
to what extent the obtained results correspond to the single expert's
assessments. On the basis of the inconsistency and discrepancy indices,
two properties of the weight deriving procedure are formulated. These
properties are proven for Eigenvector Method and Koczkodaj's Inconsistency
Index. Several estimates using Koczkodaj's Inconsistency Index for
a principal eigenvalue, Saaty's inconsistency index and the Condition
of Order Preservation are also provided. 
\end{abstract}

\section{Introduction }

The first documented uses of comparisons in pairs date back to the
thirteenth century \cite{Colomer2011rlfa}. Later, the method was
developed by \emph{Fechner} \cite{Fechner1966eop}, \emph{Thurstone}
\cite{Thurstone27aloc} and Saaty \cite{Saaty1977asmf}. The latter
proposed the Analytic Hierarchy Process \emph{(AHP)} extension to
the pairwise comparisons (herein abbreviated as\emph{ PC}) theory,
the framework allowing dealing with a large number of criteria. 

At the beginning of the twentieth century the method was used in psychometrics
and psychophysics \cite{Thurstone27aloc}. Now the method is considered
part of decision theory \cite{Saaty2005taha}. Its utility has been
confirmed in numerous examples \cite{Vaidya2006ahpa,Ho2008iahp,Liberatore2008tahp,Subramanian2012aroa}.
Despite its long existence the area still prompts researchers to further
exploration. Examples of such exploration are the \emph{Rough Set}
approach \cite{Greco2011fk}, voting systems \cite{Faliszewski2009lacv},
fuzzy \emph{PC} relation handling \cite{Mikhailov2003dpff,Yuen2013fcnp},
incomplete \emph{PC} relation \cite{Bozoki2010ooco,Fedrizzi2007ipca},
non-numerical rankings \cite{Janicki2012oapc}, nonreciprocal \emph{PC}
relation properties \cite{Fulop2012oscp}, rankings with the reference
set of alternatives \cite{Kulakowski2013ahre,Kulakowski2013hreaCoRR}
and others. Further references can be found in \cite{Smith2004aada,Ishizaka2009ahpa}. 

The aim of the \emph{PC} method is to determine the weights of objects,
so that the more important object is, the higher (or the lower) weight
it has. The method provides the user with a number of specific techniques
for deriving weights from the \emph{PC} matrix $M$ \cite{Bozoki2008osak}.
With every specific technique an appropriate inconsistency index is
associated that describes to what extent $M$ is inconsistent. The
value of an inconsistency index is perceived as a kind of quality
determinant for $M$ - the input data to the deriving weight procedure.
Following the popular adage ``garbage in, garbage out'' one could
say that when the inconsistency is high (consistency is low) the result
must be poor. Indeed, for various inconsistency indices there exist
thresholds of acceptability, above which the obtained results are
considered to be unreliable. 

However, poor and unreliable results may be also due to the deriving
method itself. One way to learn something about the heuristic procedure
is to compare how far the data on input are from the ideal input with
the resulting quality on output. In the Pairwise Comparisons (PC)
Method the input quality is determined by inconsistency indices. In
the present article the global ranking discrepancy \cite{Kulakowski2013nodi}
is proposed as a method of determining the output quality in the pairwise
comparisons method. Based on both the inconsistency index and the
global ranking discrepancy two basic properties of the weights deriving
procedure are formulated (Sec. \ref{sec:Properties-of-the}). Both
postulated properties are proven for the pair: the eigenvalue based
method and Koczkodaj's inconsistency index. The consequence of this
fact are four theorems describing the relationship between Koczkodaj's
Index and Saaty's index \cite{Saaty2013dk}, the principal eigenvalue
of the matrix $M$, the first and the second condition of order preservation
introduced by Bana e Costa and Vansninck \cite{BanaeCosta2008acao}
(Sec. \ref{sec:Properties-of-the-K-Idx}).

\section{Preliminaries }

\subsection{Pairwise comparisons method}

The input data for the \emph{PC} method is a \emph{PC} matrix $M=[m_{ij}]$,
where $m_{ij}\in\mathbb{R}_{+}$ and $i,j\in\{1,\ldots,n\}$, that
expresses a quantitative relation $R$ over the finite set of concepts
$C\overset{\textit{df}}{=}\{c_{i}\in\mathscr{C}\wedge i\in\{1,\ldots,n\}\}$.
The set $\mathscr{C}$ is a non empty universe of concepts and $R(c_{i},c_{j})=m_{ij}$,
$R(c_{j},c_{i})=m_{ji}$. The values $m_{ij}$ and $m_{ji}$ indicate
the relative importance of concepts $c_{i}$ and $c_{j}$, so that
according to the best knowledge of experts who provide the matrix
$M$ the importance of $c_{i}$ equals the importance of $c_{j}$
multiplied by factor $m_{ij}$. 
\begin{definition}
\label{def:A-matrix-recip}A matrix $M$ is said to be reciprocal
if $\forall i,j\in\{1,\ldots,n\}:m_{ij}=\frac{1}{m_{ji}}$ and $M$
is said to be consistent if $\forall i,j,k\in\{1,\ldots,n\}:m_{ij}\cdot m_{jk}\cdot m_{ki}=1$.
\end{definition}
Since the \emph{PC} matrix is usually created by humans (experts),
the information contained therein may be inconsistent. That is, there
may be a triad of values $m_{ij},m_{jk},m_{ki}$ from $M$ for which
$m_{ij}\cdot m_{jk}\cdot m_{ki}\neq1$. In other words, different
ways of estimating the value of a concept may lead to different results.
This fact leads to the concept of an inconsistency index describing
the extent to which the matrix $M$ is inconsistent. There are a number
of inconsistency indexes associated with the pairwise comparisons
deriving methods, including \emph{Eigenvector Method} \cite{Saaty1977asmf},
\emph{Least Squares Method}, \emph{Chi Squares Method} \cite{Bozoki2008osak},
\emph{Koczkodaj's} \emph{distance based inconsistency index }\cite{Koczkodaj1993ando}
and others. The two best-known indexes are defined below.
\begin{definition}
The eigenvalue based consistency index \emph{(Saaty's index)} of $n\times n$
reciprocal matrix $M$ is equal to: 
\begin{equation}
\mathcal{S}(M)=\frac{\lambda_{\textit{max}}-n}{n-1}\label{eq:Consistency_Index_AHP}
\end{equation}

where $\lambda_{\textit{max}}$ is the principal eigenvalue of $M$.
\end{definition}
The value $\lambda_{\textit{max}}\geq n$ and $\lambda_{\textit{max}}=n$
only if $M$ is consistent \cite{Saaty2013dk}. The more accurate
estimate of the value $\lambda_{\textit{max}}$ is shown in (Theorem
\ref{principal_eigen_val_estim}).
\begin{definition}
\label{def:Koczkodaj's-inconsistency-index}Koczkodaj's inconsistency
index $\mathscr{K}$ of $n\times n$ and ($n>2)$ reciprocal matrix
$M$ is equal to:

\begin{equation}
\mathscr{K}(M)=\underset{i,j,k\in\{1,\ldots,n\}}{\max}\left\{ \min\left\{ \left|1-\frac{m_{ij}}{m_{ik}m_{kj}}\right|,\left|1-\frac{m_{ik}m_{kj}}{m_{ij}}\right|\right\} \right\} \label{eq:1-koczkod_inc_idx-2}
\end{equation}

where $i,j,k=1,\ldots,n$ and $i\neq j\wedge j\neq k\wedge i\neq k$. 
\end{definition}
The result of the pairwise comparisons method is a ranking - a mapping
that assigns values to the concepts. Formally, it can be defined as
the following discrete function. 
\begin{definition}
The ranking function for $C$ (the ranking of $C$) is a function
$\mu:C\rightarrow\mathbb{R}_{+}$ that assigns to every concept from
$C\subset\mathscr{C}$ a positive value from $\mathbb{R}_{+}$. 
\end{definition}
In other words, $\mu(c)$ represents the ranking value for $c\in C$.
The $\mu$ function is usually written in the form of a vector of
weights i.e. $\mu\overset{\textit{df}}{=}\left[\mu(c_{1}),\ldots\mu(c_{n})\right]^{T}$.
One popular method assumes that $\mu$ is $\mu_{max}=\left[\mu_{\textit{max}}(c_{1}),\ldots\mu_{\textit{max}}(c_{n})\right]^{T}$
the principal eigenvector of $M$ and rescale them so that the sum
of its elements is $1$, i.e. 

\begin{equation}
\mu_{\textit{ev}}=\left[\frac{\mu_{\textit{max}}(c_{1})}{s_{\textit{ev}}},\ldots,\frac{\mu_{\textit{\textit{max}}}(c_{n})}{s_{\textit{ev}}}\right]^{T}\,\,\,\,\,\mbox{where}\,\,\,\, s_{\textit{ev}}=\underset{i=1}{\overset{n}{\sum}}\mu_{\textit{max}}(c_{i})\label{eq:eigen-value-approach}
\end{equation}

where $\mu_{\textit{ev}}$ - is the ranking function. Due to the \emph{Perron-Frobenius}
theorem \cite{Perron1907ctdm,Saaty1977asmf} $\mu_{\textit{max}}$
exists, because a real square matrix with positive entries has a unique
largest real eigenvalue such that the associated eigenvector has strictly
positive components.

\subsection{Ranking Discrepancy}

Following the \emph{PC} matrix definition their entries represent
the relative importance of concepts, so that one would expect that
for $M=[m_{ij}]$ holds that $m_{ij}=\nicefrac{\mu(c_{i})}{\mu(c_{j})}$.
Since $m_{ij}$ is the results of a subjective expert judgment, it
is subject to error. Therefore, in practice $m_{ij}$ only approximates
the ratio $\nicefrac{\mu(c_{i})}{\mu(c_{j})}$ i.e. $m_{ij}\approx\nicefrac{\mu(c_{i})}{\mu(c_{j})}$.
Let us define the difference between $m_{ij}$ and $\nicefrac{\mu(c_{i})}{\mu(c_{j})}$
formally. 
\begin{definition}
Let the ranking error $\epsilon$ be the value: 
\begin{equation}
\epsilon(i,j,\mu)\overset{\textit{df}}{=}m_{ji}\frac{\mu(c_{i})}{\mu(c_{j})}=\frac{1}{m_{ij}}\frac{\mu(c_{i})}{\mu(c_{j})}\label{eq:K_def}
\end{equation}

where $M$ is the \emph{PC} matrix and $\mu$ is the ranking function
over concepts represented by $M$. Whenever the parameter $\mu$ is
known or irrelevant to the conducted reasoning the expression $\epsilon(i,j,\mu)$
will be shortened to $\epsilon(i,j)$. 
\end{definition}
The value $\epsilon(i,j,\mu)=\epsilon(i,j)$ determines how much $m_{ji}$
- a single expert judgment differs from $\nicefrac{\mu(c_{i})}{\mu(c_{j})}$
- the ranking result. In the ideal case the expert judgment should
perfectly correspond to the ranking results. Thus, for every $i,j\in\{1,\ldots,n\}$
it should hold that $\epsilon(i,j)=1$. Unfortunately, the ranking
is usually not perfect. Therefore, depending upon whether an expert
has underestimated or overestimated the relative value of $c_{i}$
with respect to $c_{j}$, the value $\epsilon(i,j)$ may be above
or below $1$. Of course the same applies to the ranking $\mu$. In
other words it may be the case that actually the judgment given as
$m_{ij}$ is correct, whilst the ranking $\mu$ is constructed defectively.
The reciprocity of $M$ implies that the value $\epsilon(i,j)$ is
also reciprocal, i.e. $\epsilon(i,j)=\frac{1}{\epsilon(j,i)}$. Therefore,
either $\epsilon(i,j)\geq1$ or $\frac{1}{\epsilon(i,j)}\geq1$. This
observation allows one to formulate the following local ranking error
$\mathcal{E}$ definition. 
\begin{definition}
Let the local ranking discrepancy $\mathcal{E}$ be the value:

\begin{equation}
\mathcal{E}(i,j,\mu)\overset{\textit{df}}{=}\max\{\epsilon(i,j,\mu)-1,\frac{1}{\epsilon(i,j,\mu)}-1\}\label{eq:local_error-1}
\end{equation}

Whenever the parameter $\mu$ is known or irrelevant to the conducted
reasoning the expression $\mathcal{E}(i,j,\mu)$ will be shorten to
$\mathcal{E}(i,j)$. 
\end{definition}
Any other value of $\mathcal{E}(i,j)$ than $0$ means that the expert
judgement given as $m_{ij}$ differs from the ratio $\nicefrac{\mu(c_{i})}{\mu(c_{j})}$
an appropriate number of times. For instance $\mathcal{E}(i,j)=0.5$
would mean that the expert judgment $m_{ij}$ is half the time more
or half the time less than\\
 $\nicefrac{\mu(c_{i})}{\mu(c_{j})}$. In other words, in this particular
case the local discrepancy between the ranking result and the expert
judgement reaches $50\%$. 
\begin{definition}
\label{def:discrepancy_definition}Let the global ranking discrepancy
$\mathcal{D}(M,\mu)$ for the pairwise comparisons matrix $M$, and
the ranking $\mu$, be the maximal value of $\mathcal{E}(i,j,\mu)$
for $i,j=1,\ldots,n$, i.e.

\begin{equation}
\mathcal{D}(M,\mu)\overset{\textit{df}}{=}\max_{i,j\in\{1,\ldots,n\}}\mathcal{E}(i,j,\mu)\label{eq:s_loc_def-1}
\end{equation}

\end{definition}
Thus, the ranking discrepancy \cite{Kulakowski2013nodi} represents
the largest local ranking error and reveals to the users the worst
case of discrepancy in the ranking $\mu$ and the matrix $M$.

\section{Sources of discrepancy}

As it was indicated in the previous section, there can be two main
reasons why $\epsilon(i,j)\neq1$. The first is the poor quality of
judgements provided by experts. Expert estimates may be inaccurate,
flawed or disturbed in some other way%
\footnote{Some researchers argue that one cause of inconsistency (thus also
discrepancy) is the judgement scale \cite{Ishizaka2011rotm}. There
are some theoretical reasons \cite{Fulop2012oscp,FulopKS10adpo} for
using a judgement scale smaller than the F\"u{}l\"o{}p constant~i.e.~$\sqrt{\frac{1}{2}(11+5\sqrt{5})}\approx3.330191$.%
}. The method for assessing the quality of expert estimates is to determine
the input data inconsistency level, hence to calculate an inconsistency
index. Most inconsistency indices depend on both the results of paired
comparisons, and the ranking result given as a vector $\mu$ \cite{Bozoki2008osak}.
Thus, their applicability is limited to a specific ranking method.
For instance, the most popular \emph{Saaty's} inconsistency index
$\mathcal{S}$ uses the principal eigenvalue of $M$, thus, indirectly
depends on the eigenvector $\mu$. The exception is \emph{Koczkodaj's}
inconsistency index $\mathscr{K}$, defined only on the basis of $M$.
It does not depend on any particular priorities deriving method. This
property makes it universal and suitable for any weights calculation
scheme%
\footnote{A comparative analysis of both $\mathcal{S}(M)$ and $\mathscr{K}(M)$
can be found in \cite{Bozoki2008osak}.%
}. For further consideration we assume that the inconsistency index
$\textit{IC}$ depends only on the \emph{PC} matrix $M$ having in
mind that some $\textit{IC}$ makes sense only in the context of some
specific priorities deriving methods. Thus, the inconsistency index
could be written in the form:
\begin{equation}
\textit{IC}:\mathcal{M}_{\mathbb{R}_{+}}(n)\rightarrow\mathbb{R}_{+}\cup\{0\}\label{eq:inconsistency_index_abstract}
\end{equation}

where $\mathcal{M}_{\mathbb{R}_{+}}(n)$ is the set of all reciprocal
matrices $n\times n$ over $\mathbb{R}_{+}$. The inconsistency index
$\textit{IC\,}(M)$ equals $0$ when $M\in\mathcal{M}_{\mathbb{R}_{+}}(n)$
is consistent, i.e. following \cite{Saaty1977asmf}, for every triad
$m_{ik},m_{kj}$ and $m_{ij}$ of entries from $M$ holds that $m_{ik}m_{kj}=m_{ij}$.
It is assumed that the more consistent (the less inconsistent) $M$
the smaller $\textit{IC}\,(M)$. It is worth noting that inconsistency
index $\mathscr{K}$ do not use the entire set of real numbers as
a set of their values. Instead values $\mathscr{K}(M)\in[0,1)$$ $.
The work \cite{Koczkodaj2013oaoi} contains a detailed proposal (axiomatization)
of what an inconsistency index should be. 

The second reason for discrepancy is the way in which the ranking
was created. In other words, the ranking discrepancy may be a side
effect of the ranking procedure. If that is so, it is worth considering
what properties should meet the priorities deriving method. To answer
this question let us define the ranking method more formally. 
\begin{definition}
Let the priority deriving procedure for the pairwise comparisons method
be represented by the following mapping $P$: 
\begin{equation}
P:\mathcal{M}_{\mathbb{R}_{+}}(n)\rightarrow L\label{eq:priority_deriving_method_abstract}
\end{equation}

where $L=\left\{ \mu\,|\,\mu:C\rightarrow\mathbb{R}_{+}\right\} $
is the space of ranking functions over $C=\{c_{1},\ldots,c_{n}\}$,
and $\mathcal{M}_{\mathbb{R}_{+}}(n)$ is the set of all reciprocal
matrices $n\times n$ over $\mathbb{R}_{+}$. Let the set of all priority
deriving procedures be denoted by $\mathcal{P}$. The specific ranking
designated by $P(M)$ will be denoted as $\mu_{P(M)}$ or $\mu_{P}$
(if $M$ is known), i.e. $P(M)=\mu_{P(M)}=\mu_{P}$. Of course, although
$P$ is written with the help of the functional notation (which may
suggest that $P$ is a function) in practice it might be implemented
as an arbitrarily complicated procedure. 
\end{definition}

\section{Properties of the Priority Deriving Procedure\label{sec:Properties-of-the}}

The principal assumption of the pairwise comparisons method is that
for $M=[m_{ij}]$ holds $m_{ij}$ reflects the relative importance
of concepts. Hence the situation in which $m_{ij}\neq\nicefrac{\mu(c_{i})}{\mu(c_{j})}$,
can be caused only by the inconsistency of $M$ regardless of how
$\mu$ has been chosen. Conversely, for a consistent $M$ one would
expect that $m_{ij}=\nicefrac{\mu(c_{i})}{\mu(c_{j})}$ for all $i,j\in\{1,\ldots,n\}$
and for any $\mu\in L$ obtained by the procedure $P\in\mathcal{P}$.
In other words it is required that for the regular ranking procedure
$P$ the existence of inconsistency in $M$ is the direct cause of
the existence of discrepancy (Def. \ref{def:discrepancy_definition}).
This common-sense postulate can be written as follows:
\begin{proposition}
\label{proposition1}The priority deriving procedure $P\in\mathcal{P}$
is said to be regular if 
\begin{equation}
\textit{IC}\,(M)=0\Rightarrow\mathcal{D}(M,\mu_{P(M)})=0\label{eq:first_proposition}
\end{equation}

for the given (applicable) inconsistency index $\textit{IC}$ and
the PC matrix $M\in\mathcal{M}_{\mathbb{R}_{+}}(n)$. 
\end{proposition}
The proposition formulated above can be seen as a kind of relationship
between $\textit{IC}$ and $P$. Indeed, when considering an $\textit{IC}$
that depends only from $M$ (e.g. Koczkodaj's inconsistency index
$\mathcal{K}$), since $\mathcal{D}$ is defined and fixed, the first
proposition describes an inherent property of $P$. 

The matrix $M$ is usually created as the result of the hard work
of experts in the field of the relation $R$. Individuals, including
experts, are often inconsistent in their judgements \cite{Roberts1985mtwa,Saaty2013dk}.
However, there is a level of tolerable inconsistency beyond which
judgments would appear to be uninformed, random, or arbitrary \cite{Saaty2013dk}.
Thus, the question arises of how to deal with excessive inconsistency.
The literature \cite{Ho2008iahp,Saaty2013dk} advises revising the
matrix $M$ so that the new version of the matrix is more (sufficiently)
consistent. Another method is to find the closest (in the geometrical
sense) consistent approximation of $M$ \cite{Koczkodaj2010odbi}.
The purpose of the inconsistency reduction is to make the values $m_{ij}$
and $m{}_{ik}m_{kj}$ closer to each other, i.e. (having in mind that
$m_{ik}\approx\nicefrac{\mu(c_{i})}{\mu(c_{k})}$ and $m_{kj}\approx\nicefrac{\mu(c_{k})}{\mu(c_{j})}$)
minimizing the discrepancy. Hence, it is natural to expect that decreasing
inconsistency leads to a discrepancy reduction. This postulate can
be formulated as follows:
\begin{proposition}
\label{proposition2}It is said that the priority deriving procedure
$P\in\mathcal{P}$ follows the inconsistency if 
\begin{equation}
\textit{IC\,}(M)\rightarrow0\Rightarrow\mathcal{D}(M,\mu_{P(M)})\rightarrow0\label{eq:second_proposition_a}
\end{equation}

and there exists a reasonably small $0<\kappa<\textit{IC}(M)$ such
that: 
\begin{equation}
\textit{IC}\,(M)\geq\textit{IC}\,(M')+\kappa\Rightarrow\mathcal{D}(M,\mu_{P(M)})>\mathcal{D}(M',\mu_{P(M)})\label{eq:second_proposition_b}
\end{equation}

for the given (applicable) inconsistency index $\textit{IC}$, where
$M$ is the PC matrix from $\mathcal{M}_{\mathbb{R}_{+}}(n)$.
\end{proposition}
This proposition \ref{proposition2} meets the expectations that improving
the quality of input data eventually brings the expected results.
In other words if the inconsistency index for the given matrix $M$
is appropriately reduced then there is a guarantee that discrepancy
will also be reduced. 

Both are defined in the context of any fixed inconsistency index $\textit{IC}$.
Therefore, when determining if the procedure has the given property
a suitable index $\textit{IC}$ needs to be taken into account.

\section{Properties of the Eigenvalue based Priority Deriving Procedure\label{sec:Properties-of-the-Eigenvalue-based-PDP}}

Of course the value $\mathcal{D}(M,\mu_{P(M)})=0$ is optimal from
the perspective of the decision maker (there is no doubt what to choose
since every single expert judgement perfectly matches the ranking
result). Since, decreasing the inconsistency $\textit{IC}(M)$ entails
decreasing\\
$\mathcal{D}(M,\mu_{P(M)})$ down to $0$, decreasing $\textit{IC}(M)$
makes sense. Therefore, meeting the properties presented in (Sec.
\ref{sec:Properties-of-the}) provides arguments for decreasing inconsistency
in $M$. It also provides strong arguments for using the given pair
$\textit{IC}$ and $P$ together. In this section the most popular
eigenvalue based priority deriving procedure $P_{\textit{ev}}$ and
two applicable inconsistency indices $\mathcal{K}$ and $\mathcal{S}$
are examined. Both properties, as proposed in (Sec. \ref{sec:Properties-of-the}),
are confirmed for the pair $(\mathcal{K},P_{\textit{ev}})$. For the
pair $(S,P_{\textit{ev}})$ the first property is proven. 

To demonstrate that for the first pair $(\mathcal{K},P_{\textit{ev}})$
both properties hold, and the one property implies the other, first
let us prove the following auxiliary theorem. 
\begin{lemma}
\label{lemma1}For every pairwise comparisons matrix $M\in\mathcal{M}_{\mathbb{R}_{+}}(n)$
and the eigenvalue based pairwise comparisons procedure $P_{\textit{ev}}$
it holds that: 
\begin{equation}
\alpha\leq\epsilon(i,j,\mu_{P_{\textit{ev}}(M)})\leq\frac{1}{\alpha}\label{eq:lemma_eq}
\end{equation}

where $\alpha\overset{\textit{df}}{=}1-\mathcal{K}(M)$. \end{lemma}
\begin{proof}
Following the equation \ref{eq:1-koczkod_inc_idx-2}, Koczkodaj's
distance inconsistency index $\mathscr{K}(M)$, in short $\mathscr{K}$,
means that the maximal local inconsistence for some maximal triad
$m_{pq},m_{qr}$ and $m_{pr}$ is $\mathscr{K}$. Thus, in the case
of any triad in form of $m_{ik},m_{kj},m_{ij}$ it must hold that:

\begin{equation}
\mathscr{K}\geq\min\left\{ \left|1-\frac{m_{ij}}{m_{ik}m_{kj}}\right|,\left|1-\frac{m_{ik}m_{kj}}{m_{ij}}\right|\right\} \label{eq:the_eq_2}
\end{equation}

This means that either: 

\begin{equation}
m_{ij}\leq m_{ik}m_{kj}\,\,\text{implies}\,\,\mathscr{K}\geq1-\frac{m_{ij}}{m_{ik}m_{kj}}\label{eq:the_eq_3}
\end{equation}

or

\begin{equation}
m_{ik}m_{kj}\leq m_{ij}\,\,\text{implies}\,\,\mathscr{K}\geq1-\frac{m_{ik}m_{kj}}{m_{ij}}\label{eq:the_eq_4}
\end{equation}

is true. Let us denote $\alpha\overset{\textit{df}}{=}1-\mathscr{K}$. The
statements above can then be written in the form: 

\begin{equation}
m_{ij}\leq m_{ik}m_{kj}\,\,\text{implies}\,\, m_{ij}\geq\alpha\cdot m_{ik}m_{kj}\label{eq:the_eq_5}
\end{equation}

\begin{equation}
m_{ik}m_{kj}\leq m_{ij}\,\,\text{implies}\,\,\frac{1}{\alpha}\cdot m_{ik}m_{kj}\geq m_{ij}\label{eq:the_eq_6}
\end{equation}

Since $\alpha\leq1$ , both cases (\ref{eq:the_eq_5}) and (\ref{eq:the_eq_6})
lead independently to the common conclusion that: 
\begin{equation}
\alpha\cdot m_{ik}m_{kj}\leq m_{ij}\leq\frac{1}{\alpha}m_{ik}m_{kj}\label{eq:triad_estim}
\end{equation}

for every $i,j,k\in\{1,\ldots,n\}$. 

On the other hand, following the $P_{\textit{ev}}$ procedure \cite{Saaty2013dk}
the vector $\mu\overset{\textit{df}}{=}\mu_{P_{\textit{ev}}(M)}$
is the principal eigenvector of $M$. Thus, it satisfies the equation:
\begin{equation}
M\mu=\lambda_{\textit{max}}\mu\label{eq:eigen_eq}
\end{equation}

where $\lambda_{\textit{max}}$ is the principal eigenvalue of $M$.
The i-th equation of (\ref{eq:eigen_eq}) has the form: 
\begin{equation}
m_{i1}\mu(c_{1})+\ldots+m_{in}\mu(c_{n})=\lambda_{\textit{max}}\cdot\mu(c_{i})\label{eq:single_eigen_eq}
\end{equation}

In other words $\epsilon$ can be written as:
\begin{equation}
\epsilon(i,j,\mu)\overset{\textit{df}}{=}\frac{1}{m_{ij}}\cdot\frac{\mu(c_{i})}{\mu(c_{j})}=\frac{1}{m_{ij}}\cdot\frac{m_{i1}\mu(c_{1})+\ldots+m_{in}\mu(c_{n})}{m_{j1}\mu(c_{1})+\ldots+m_{jn}\mu(c_{n})}\label{eq:kappa_eq2}
\end{equation}

Applying (\ref{eq:triad_estim}) to the left side of (\ref{eq:single_eigen_eq}),
we obtain
\begin{equation}
m_{i1}\mu(c_{1})+\ldots+m_{in}\mu(c_{n})\leq\frac{1}{\alpha}\left(m_{ij}m_{j1}\mu(c_{1})+\ldots+m_{ij}m_{jn}\mu(c_{n})\right)\label{eq:final_estim_1}
\end{equation}

and accordingly: 
\begin{equation}
\alpha\left(m_{ij}m_{j1}\mu(c_{1})+\ldots+m_{ij}m_{jn}\mu(c_{n})\right)\leq m_{i1}\mu(c_{1})+\ldots+m_{in}\mu(c_{n})\label{eq:final_estim_2}
\end{equation}

Thus (\ref{eq:final_estim_1}) means that: 
\begin{equation}
\frac{1}{m_{ij}}\cdot\frac{\mu(c_{i})}{\mu(c_{j})}\leq\frac{1}{m_{ij}}\cdot\left(\frac{1}{\alpha}\cdot\frac{m_{ij}\left(m_{j1}\mu(c_{1})+\ldots+m_{jn}\mu(c_{n})\right)}{m_{j1}\mu(c_{1})+\ldots+m_{jn}\mu(c_{n})}\right)=\frac{1}{\alpha}\label{eq:final_estim_3}
\end{equation}

and accordingly (\ref{eq:final_estim_2}) implies that

\begin{equation}
\alpha=\frac{1}{m_{ij}}\cdot\left(\alpha\cdot\frac{m_{ij}\left(m_{j1}\mu(c_{1})+\ldots+m_{jn}\mu(c_{n})\right)}{m_{j1}\mu(c_{1})+\ldots+m_{jn}\mu(c_{n})}\right)\leq\frac{1}{m_{ij}}\cdot\frac{\mu(c_{i})}{\mu(c_{j})}\label{eq:final_estim_4}
\end{equation}

Both the above inequalities lead to the conclusion that: 
\begin{equation}
\alpha\leq\epsilon(i,j,\mu)\leq\frac{1}{\alpha}\label{eq:final_assert}
\end{equation}

which is the desired assertion. 
\end{proof}
Equipped with the Lemma \ref{lemma1}, we can easily prove the properties
of $P_{\textit{ev}}$ with respect to the inconsistency index $\mathcal{K}$.
\begin{theorem}
The eigenvalue based pairwise comparisons procedure $P_{\textit{ev}}$
is regular and follows the inconsistency with respect to Koczkodaj's
inconsistency index $\mathcal{K}$. \end{theorem}
\begin{proof}
Let $\mu\overset{\textit{df}}{=}P_{\textit{ev}}(M)$ be fixed and
known. Thus, from Lemma \ref{lemma1} 
\begin{equation}
\alpha-1\leq\epsilon(i,j)-1\leq\frac{1}{\alpha}-1\label{eq:theo_eq1}
\end{equation}

for all $i,j\in\{1,\ldots,n\}$. The same applies to $\epsilon(j,i)$.
Thus, due to the reciprocity of $\epsilon$ also 
\begin{equation}
\alpha-1\leq\frac{1}{\epsilon(i,j)}-1\leq\frac{1}{\alpha}-1\label{eq:eq:theo_eq2}
\end{equation}

It is easy to observe that 
\begin{equation}
0\leq\max\{\epsilon(i,j)-1,\frac{1}{\epsilon(i,j)}-1\}\leq\frac{1}{\alpha}-1\label{eq:theo_eq3}
\end{equation}

In other words: 
\begin{equation}
0\leq\mathcal{E}(i,j)\leq\frac{1}{\alpha}-1\label{eq:theo_eq4}
\end{equation}

Since the above (\ref{eq:theo_eq4}) is true for all $i,j\in\{1,\ldots,n\}$
then, due to (Def. \ref{def:discrepancy_definition}), holds that:
\begin{equation}
0\leq\mathcal{D}(M,\mu)\leq\frac{1}{\alpha}-1\label{eq:theo_estim}
\end{equation}

Since $\alpha\overset{\textit{df}}{=}1-\mathscr{K}(M)$, then 
\begin{equation}
\mathscr{K}(M)\rightarrow0\Rightarrow\alpha\rightarrow1\Rightarrow\left(\frac{1}{\alpha}-1\right)\rightarrow0\label{eq:theo_converg}
\end{equation}

Thus, due to (\ref{eq:theo_estim}) 
\begin{equation}
\mathscr{K}(M)\rightarrow0\,\,\,\,\,\text{implies}\,\,\,\,\,\mathcal{D}(M,\mu)\rightarrow0\label{eq:concl_theo_1}
\end{equation}

which proves the first postulate (\ref{eq:second_proposition_a})
of the second property. The estimation (\ref{eq:theo_estim}) also
allows one to prove the first property. Simply, for $\mathscr{K}(M)=0$
means that $\left(\frac{1}{\alpha}-1\right)=0$ therefore 
\begin{equation}
0\leq\mathcal{D}(M,\mu)\leq0\,\,\,\,\,\text{implies}\,\,\,\,\,\mathcal{D}(M,\mu)=0\label{eq:concl_theo_2}
\end{equation}

An important question during the process of inconsistency reduction
is whether the improvement is large enough to have an actual impact
on the final ranking. Let us consider two matrices $M$ and $M'$
where the second one was obtained from the first one as the result
of an inconsistency reduction process (judgment revision \cite{Ho2008iahp},
inconsistency reduction algorithm \cite{Koczkodaj2010odbi}). The
question of how much (at least) the inconsistency should be reduced
to be important for the all ranked concepts boils down to the question
of the reasonably small $\kappa\in\mathbb{R}_{+}\backslash\{0\}$.
To answer this question look at (\ref{eq:theo_estim}). In particular
the inequality (\ref{eq:second_proposition_b}) is met if $\mathcal{K}(M')$
is so that 
\begin{equation}
\mathcal{D}(M',\mu)\leq\left(\frac{1}{1-\mathcal{K}(M')}-1\right)<\mathcal{D}(M,\mu)\label{eq:step_ineq_2}
\end{equation}

Thus, in particular we demand that 
\begin{equation}
\mathcal{K}(M')<1-\frac{1}{\mathcal{D}(M,\mu)+1}\label{eq:step_ineq_3}
\end{equation}

hence,

\begin{equation}
\mathcal{K}(M)-\mathcal{K}(M')>\mathcal{K}(M)+\frac{1}{\mathcal{D}(M,\mu)+1}-1\label{eq:step_ineq_4}
\end{equation}

Therefore, the suitable $\kappa$ candidate is: 
\begin{equation}
\kappa\overset{\textit{df}}{=}\mathcal{K}(M)+\frac{1}{\mathcal{D}(M,\mu)+1}-1\label{eq:step_ineq_5}
\end{equation}

It is clear that due to the (\ref{eq:step_ineq_4}) $\kappa<\mathcal{K}(M)<1$.
It is also $0\leq\kappa$ ($0<\kappa$, when $\mathcal{K}(M)>0)$.
The fact that $\kappa$ is not negative is a simple conclusion from
(\ref{eq:theo_estim}). It holds that: 
\begin{equation}
\mathcal{D}(M,\mu)\leq\frac{1}{1-\mathcal{K}(M)}-1\label{eq:step_ineq_6}
\end{equation}

Thus, it is easy to see that: 
\begin{equation}
\mathcal{K}(M)+\frac{1}{\mathcal{D}(M,\mu)+1}\geq1\label{eq:step_ineq_7}
\end{equation}

hence the right side of the equation (\ref{eq:step_ineq_5}) is greater
or equal to %
\footnote{Indeed, for the fully consistent $M_{1}$ ($\mathcal{K}(M_{1})=\mathcal{D}(M_{1},\mu)=0$),
an $\kappa$ candidate is $0$%
} $0$. Therefore, due to the way in which it was constructed, $\kappa$
as proposed in (\ref{eq:step_ineq_5}) satisfies the conditions of
the second property.
\end{proof}
Lemma \ref{lemma1} immediately implies another useful assertion that
allow us to use Koczkodaj's inconsistency to relatively estimate one
entry of the principal eigenvector by another. This leads to the following
corollary. 
\begin{corollary}
Every entry in the principal eigenvector bounds each other according
to the following inequality

\begin{equation}
\alpha m_{ij}\mu(c_{j})\leq\mu(c_{i})\leq\frac{1}{\alpha}m_{ij}\mu(c_{j})\label{eq:lemma_eq-1-1}
\end{equation}

where $\alpha\overset{\textit{df}}{=}1-\mathcal{K}(M)$ and $M\mu=\lambda_{\textit{max}}\mu$. 
\end{corollary}
The situation of discrepancy in which the expert judgment $m_{ij}$
differs from the ranking result $\nicefrac{\mu(c_{i})}{\mu(c_{j})}$
is not comfortable for people interested in the ranking results. The
high discrepancy may be the cause of complaints, doubts as to the
result, lack of a sense of justice and so on. In other words, the
high discrepancy may cause customer satisfaction with the results
of the ranking to be low. Reversely, when the discrepancy is small,
customers are likely to be more satisfied than when the discrepancy
is high. Although some customers will probably never be totally satisfied
with the ranking (this particularly applies to those who support the
concept of lower values $\mu$) in general it can be assumed that
the higher the discrepancy the lower the satisfaction, and conversely,
the lower the discrepancy the higher the satisfaction. Of course,
the most attention-grabbing and emotive situations are those of the
highest discrepancy. Reducing the inconsistency $\mathcal{K}(M)$
by (at least) $\kappa$ allows each time to reduce the most severe
and unsatisfying cases of discrepancy. With help comes the second
property that leads to the following conclusion. 
\begin{corollary}
For the eigenvalue based pairwise comparisons procedure $P_{\textit{ev}}$,
the inconsistency index $\mathcal{K}$, and the given PC matrix $M$,
it is recommended to reduce the inconsistency $\mathcal{K}(M)$ by
such $\kappa$ that guarantees reduction of the discrepancy $\mathcal{D}(M,\mu)$.
One possible $\kappa$ is: 
\begin{equation}
\kappa\overset{\textit{df}}{=}\mathcal{K}(M)+\frac{1}{\mathcal{D}(M,\mu)+1}-1\label{eq:step_ineq_9}
\end{equation}
\end{corollary}
\begin{theorem}
The eigenvalue based pairwise comparisons procedure $P_{\textit{ev}}$
is regular with respect to Saaty's inconsistency index $\mathcal{S}$. \end{theorem}
\begin{proof}
The proof of the first property immediately results from \cite{Saaty1977asmf},
i.e. it is easy to see that $\mathcal{S}(M)=0$ if and only if $\lambda_{\textit{max}}=n$.
Hence, according to \cite[theorem 1]{Saaty1977asmf}, the equality
$\lambda_{\textit{max}}=n$ implies that $m_{ij}=\nicefrac{\mu(c_{i})}{\mu(c_{j})}$
for every $i,j\in\{1,\ldots,n\}$. Thus, every $\mathcal{E}(i,j)=0$,
hence, $\mathcal{D}(M,\mu)=0$.
\end{proof}

\section{Properties of Koczkodaj's inconsistency index\label{sec:Properties-of-the-K-Idx}}

The theorems proven in (Sec. \ref{sec:Properties-of-the-Eigenvalue-based-PDP})
allow the use of $\mathcal{K}(M)$ for the effective estimation of
other quantities. Some of the inequalities involving $\mathcal{K}(M)$
will be presented below in the form of appropriate assertions. 
\begin{theorem}
For any PC matrix $M$ holds that

\begin{equation}
\mathcal{S}(M)\leq\frac{1}{1-\mathcal{K}(M)}-1\label{eq:theo_thesis_3}
\end{equation}

where $\mathcal{S}(M),\mathcal{K}(M)$ are Saaty's and Koczkodaj's
inconsistency indices respectively.\end{theorem}
\begin{proof}
It is known \cite{Saaty1977asmf,Kulakowski2013nodi} that the following
inequality holds:

\begin{equation}
\frac{1}{\left(n-1\right)}\sum_{i=1,i\neq j}^{n}\left(\epsilon(i,j)-1\right)=\mathcal{S}(M)\label{eq:SM_proof_3}
\end{equation}

On the other hand from the Lemma \ref{lemma1} ($\mu=\mu_{P_{\textit{ev}}(M)}$
is omitted) we know that

\begin{equation}
\alpha\leq\epsilon(i,j)\leq\frac{1}{\alpha}\label{eq:theo3_2}
\end{equation}

hence, 

\begin{equation}
\alpha-1\leq\epsilon(i,j)-1\leq\frac{1}{\alpha}-1\label{eq:theo3_3}
\end{equation}

Therefore,
\begin{equation}
\left(n-1\right)\left(\alpha-1\right)\leq\sum_{i=1,i\neq l}^{n}\left(\epsilon(l,i)-1\right)\leq(n-1)\left(\frac{1}{\alpha}-1\right)\label{eq:theo3_4}
\end{equation}

Then by dividing both sides by $(n-1)$ and applying (\ref{eq:SM_proof_3})
we get:

\begin{equation}
\alpha-1\leq\mathcal{S}(M)\leq\frac{1}{\alpha}-1\label{eq:theo3_5}
\end{equation}

which is the desired assertion.
\end{proof}
The above theorem allows one to estimate the maximum value of Saaty's
inconsistency index using $\mathcal{K}$. In particular it is easy
to see that $\mathcal{K}(M)=0.09(09)$ implies that $\mathcal{S}(M)<0.1$.
In other words a Koczkodaj inconsistency $\mathcal{K}$ smaller than
$0.909(09)$ guarantees meeting the consistency criteria proposed
by Saaty \cite{Saaty2013dk}. Hence every\emph{ PC }matrix $M$ for
which $\mathcal{K}(M)<0.90909$ is also consistent enough in Saaty's
sense. Of course, in such a case due to the Lemma \ref{lemma1} we
can also expect some regularity in discrepancy among different paired
comparisons. 

The above theorem immediately suggest the estimation for the principal
eigenvalue of $M$. 
\begin{theorem}
\label{principal_eigen_val_estim}For any PC matrix $M$ its principal
eigenvalue $\lambda_{\textit{max}}$ is bounded as follows: 

\begin{equation}
\left(n-1\right)\left(\alpha-1\right)+n\leq\lambda_{\textit{max}}\leq(n-1)\left(\frac{1}{\alpha}-1\right)+n\label{eq:theo_4_1}
\end{equation}

where $\alpha\overset{\textit{df}}{=}1-\mathcal{K}(M)$. \end{theorem}
\begin{proof}
Due to (Def. \ref{eq:Consistency_Index_AHP}) and (\ref{eq:theo3_5})
we obtain
\begin{equation}
\alpha-1\leq\frac{\lambda_{\textit{max}}-n}{n-1}\leq\frac{1}{\alpha}-1\label{eq:theo_4_2}
\end{equation}

thus 

\begin{equation}
\left(n-1\right)\left(\alpha-1\right)+n\leq\lambda_{\textit{max}}\leq(n-1)\left(\frac{1}{\alpha}-1\right)+n\label{eq:theo_4_3}
\end{equation}

which is the desired inequality.
\end{proof}
Since $-1<(\alpha-1)\leq0$, then $-(n-1)<(n-1)(\alpha-1)$. Therefore,
$0<\left(n-1\right)\left(\alpha-1\right)+n$, which confirms the well
known fact \cite{Saaty1977asmf} that $\lambda_{\textit{max}}$ is
positive. Since the above theorem shows that Koczkodaj's inconsistency
index $\mathcal{K}$ can be used to estimate the principal eigenvalue,
the uses of $\mathcal{K}$ extend beyond the \emph{PC} method. 

The theorems proven in \cite{Kulakowski2013nodi} allow for the formulation
of the more general Conditions of Order Preservation (See appendix
\ref{sec:Conditions-of-Order}) initially introduced by \emph{Bana
e} \emph{Costa} and\emph{ Vansnick }\cite{BanaeCosta2008acao}. Thus,
the first \emph{POP (the preservation of order preference condition)}
leads to the following new theorem. 
\begin{theorem}
\label{COP-theor-1}For the eigenvalue based pairwise comparisons
procedure $P_{\textit{ev}}$, the PC matrix $M=[m_{ij}]$, expressing
the quantitative relationships $R$ between concepts $c_{1},\ldots,c_{n}\in C$,
and the ranking $\mu_{P_{ev}(M)}=\mu$, holds that 
\begin{equation}
m_{ij}>\frac{1}{1-\mathcal{K}(M)}\,\,\,\,\,\mbox{implies}\,\,\,\,\,\mu(c_{i})>\mu(c_{j})\label{eq:cop-proof_2}
\end{equation}
\end{theorem}
\begin{proof}
From \cite[Theorem 2]{Kulakowski2013nodi} it holds that 
\begin{equation}
\left\{ \left(\mathcal{D}(M,\mu)\leq\delta\right)\Rightarrow\left(m_{ij}>\delta+1\right)\right\} \Rightarrow\left\{ \left(m_{ij}>1\right)\Rightarrow\left(\mu(c_{i})>\mu(c_{j})\right)\right\} \label{eq:cross_theo_2}
\end{equation}

for some $\delta\in\mathbb{R}_{+}$. Due to the (\ref{eq:theo_estim})
the right side of (\ref{eq:cross_theo_2}) is 
\begin{equation}
\left(\mathcal{D}(M,\mu)\leq\frac{1}{1-\mathcal{K}(M)}-1\right)\Rightarrow\left(m_{ij}>\frac{1}{1-\mathcal{K}(M)}\right)\label{eq:cop_eq3}
\end{equation}

and since $\frac{1}{1-\mathcal{K}(M)}>1$ then also 
\begin{equation}
m_{ij}>\frac{1}{1-\mathcal{K}(M)}\,\,\,\,\,\text{implies}\,\,\,\,\,\mu(c_{i})>\mu(c_{j})\label{eq:cop_eq_4}
\end{equation}

which is the desired assertion.
\end{proof}
Similarly the second \emph{POIP} (\emph{the preservation of order
of intensity of preference condition}) leads to the new interesting
theorem. 
\begin{theorem}
\label{COP-theor-2}For the eigenvalue based pairwise comparisons
procedure $P_{\textit{ev}}$, the PC matrix $M=[m_{ij}]$, expressing
the quantitative relationships $R$ between concepts $c_{1},\ldots,c_{n}\in C$,
and the ranking $\mu_{P_{ev}(M)}=\mu$, holds that: 
\begin{equation}
\frac{m_{ij}}{m_{kl}}>\left(\frac{1}{1-\mathcal{K}(M)}\right)^{2}\,\,\,\,\,\mbox{implies}\,\,\,\,\,\frac{\mu(c_{i})}{\mu(c_{j})}>\frac{\mu(c_{k})}{\mu(c_{l})}\label{eq:cop-proof_2-1}
\end{equation}
\end{theorem}
\begin{proof}
From \cite[Theorem 3]{Kulakowski2013nodi} holds that 
\begin{equation}
\left\{ \left(\mathcal{D}(M,\mu)\leq\delta\right)\Rightarrow\frac{m_{ij}}{m_{kl}}>\left(\delta+1\right)^{2}\right\} \Rightarrow\left\{ m_{ij}>m_{kl}>1\Rightarrow\frac{\mu(c_{i})}{\mu(c_{j})}>\frac{\mu(c_{k})}{\mu(c_{l})}\right\} \label{eq:cross_theo_3}
\end{equation}

for some $\delta\in\mathbb{R}_{+}$. Due to the (\ref{eq:theo_estim})
the right side of (\ref{eq:cross_theo_3}) is: 
\begin{equation}
\left(\mathcal{D}(M,\mu)\leq\frac{1}{1-\mathcal{K}(M)}-1\right)\Rightarrow\frac{m_{ij}}{m_{kl}}>\left(\frac{1}{1-\mathcal{K}(M)}-1+1\right)^{2}\label{eq:cross_theo_4}
\end{equation}

and since $\frac{1}{1-\mathcal{K}(M)}>1$ then also

\begin{equation}
\frac{m_{ij}}{m_{kl}}>\left(\frac{1}{1-\mathcal{K}(M)}\right)^{2}\Rightarrow\frac{\mu(c_{i})}{\mu(c_{j})}>\frac{\mu(c_{k})}{\mu(c_{l})}\label{eq:cross_theo_5}
\end{equation}

which is the desired assertion.
\end{proof}
The above two theorems show a direct relationship between the level
of inconsistency and the conditions of order preservations. They may
be used as a quick criterion for assessing whether in a certain case
a given condition will be met. The theorems show that the lower the
inconsistency the easier the left sides of (\ref{eq:cop-proof_2})
and (\ref{eq:cop-proof_2-1}) could be satisfied. Thus, in practice
the more consistent the \emph{PC} matrix the more often \emph{POP}
and \emph{POIP} conditions for randomly selected matrices are satisfied.

\section{Summary}

This paper presents a new perspective on the pairwise comparisons
method. Following the proposed approach the weight deriving method
is a heuristic procedure that transforms the input (a \emph{PC} matrix)
into the output (a weight vector). Hence, after determining the quality
of the input and the quality of the output it is possible to discuss
the quality of the weight deriving method. As the output quality indicator
global ranking discrepancy \cite{Kulakowski2013nodi} has been proposed.
Following the \emph{PC} method theory the input quality is determined
with the help of an inconsistency index. With the help of these two
indices two properties of a good weight deriving method have been
proposed. The first stipulates that when a \emph{PC} matrix is consistent
then there should be no discrepancy between expert judgments and the
ranking results. The second one refers to the reasonableness of reducing
inconsistency. It requires that the appropriately significant decrease
in inconsistency (if it is greater than $0$) always leads to a decrease
in discrepancy. In other words even if it is not possible to reduce
inconsistency to $0$, making a significant inconsistency reduction
must result in a discrepancy reduction. The properties were proven
for \emph{Koczkodaj's} inconsistency index $\mathcal{K}$ and the
\emph{Eigenvector} \emph{based} method $P_{\textit{ev}}$. The presented
reasoning results in four further claims revealing relationships between
$\mathcal{K}$ and the principal eigenvalue, \emph{Saaty's} inconsistency
index and two conditions of order preservation. 

The proposed properties are general, and thus relate to each pairwise
comparisons weight deriving method and each applicable inconsistency
index. Therefore, despite the fact that the presented approach allows
showing the relationship between $\mathcal{K}$ and $P_{\textit{ev}}$
from a new perspective, many questions, especially regarding the other
weight deriving methods and the inconsistency indices, remain. The
answers to these questions may bring new interesting results, which
allow us to enrich our understanding of the \emph{PC} method.

\section*{Acknowledgements }

I would like to thank Prof. Antoni Lig\k{e}za for his and insightful
comments, constant support, and reading the first version of this
work. Special thanks are due to Dan Swain for his editorial help.

\bibliographystyle{plain}
\bibliography{papers_biblio_reviewed}

\appendix

\section{Conditions of Order Preservations\label{sec:Conditions-of-Order}}

In \cite{BanaeCosta2008acao} \emph{Bana e} \emph{Costa} and \emph{Vansnick}
formulate two (COP) conditions of order preservations. The first,
\emph{the preservation of order preference condition} \emph{(POP}),
claims that the ranking result in relation to the given pair of concepts
$(c_{i},c_{j})$ should not break with the expert judgement, i.e.
for a pair of concepts $c_{1},c_{2}\in C$ such that $c_{1}$ dominates
$c_{2}$ i.e. $m_{1,2}>1$ it should hold that:

\begin{equation}
\mu(c_{1})>\mu(c_{2})\label{eq:6-cop-qualitative-cond-1}
\end{equation}

The second one \emph{the preservation of order of intensity of preference
condition }(\emph{POIP),} claims that if $c_{1}$ dominates $c_{2}$,
more than $c_{3}$ dominates $c_{4}$ (for $c_{1},\ldots,c_{4}\in C$),
i.e. if additionally $m_{3,4}>1$ and $m_{1,2}>m_{3,4}$ then also

\begin{equation}
\frac{\mu(c_{1})}{\mu(c_{2})}>\frac{\mu(c_{3})}{\mu(c_{4})}\label{eq:8-eq:cop-quantitative-cond}
\end{equation}

\end{document}